\documentclass[11pt]{article}
\usepackage{fullpage}
\usepackage{amsfonts}
\usepackage{algorithmic}
\usepackage[boxed]{algorithm}
\usepackage{amsmath,amsthm,amssymb}
\usepackage{latexsym}
\usepackage{epic}
\usepackage{epsfig}
\usepackage{amscd}
\usepackage{url}
\usepackage{verbatim}
\usepackage{setspace}

\newtheorem{theorem}{Theorem}[section]

\newtheorem{lemma}[theorem]{Lemma}
\newtheorem{proposition}[theorem]{Proposition}

\newtheorem{definition}[theorem]{Definition}
\newtheorem{remark}[theorem]{Remark}

\newcommand{\sm}{\setminus}

\def\ex{\qopname\relax n{E}}
\def\min{\qopname\relax n{min}}
\def\max{\qopname\relax n{max}}

\newcommand{\expect}[2][]{\ex_{#1} [#2]}

\def\L{\mathcal{L}}

\def\sse{\subseteq}

\newcommand{\eat}[1]{}

\newenvironment{lp*}{\begin{equation*}  \begin{array}{lll}}{\end{array}\end{equation*}}

\newcommand{\alloc}{X}
\newcommand{\allocs}{{\mathbf \alloc}}

\newcommand{\alloci}[1][i]{{\alloc_{#1}}}
\newcommand{\price}{P}
\newcommand{\prices}{{\mathbf \price}}
\newcommand{\pricei}[1][i]{{\price_{#1}}}
\newif\ifshortver
\shortverfalse

\newcommand{\ShortLong}[2]{\ifshortver{#1}\else{#2}\fi}

\title{Combinatorial Auctions with Restricted Complements}

\author{Ittai Abraham\\
Microsoft Research, Silicon Valley\\
{\tt ittaia@microsoft.com} \and
Moshe Babaioff\\
Microsoft Research, Silicon Valley\\
{\tt moshe@microsoft.com} \vspace{0.3cm} \and
Shaddin Dughmi\\
Microsoft Research, Redmond\\
{\tt shaddin@microsoft.com} \and
Tim Roughgarden\\
Stanford University\\
{\tt tim@theory.stanford.edu} }

\begin{document}
\maketitle

\begin{abstract}
Complements between goods --- where one good takes on added value in
the presence of another --- have been a thorn in the side of algorithmic
mechanism designers.
On the one hand, complements are common in the
standard motivating applications for combinatorial auctions, like
spectrum license auctions.  On the other, welfare maximization in
the presence of complements is notoriously difficult, and this
intractability has stymied theoretical progress in the area.
For example, there are no known positive results for combinatorial
auctions in which bidder valuations are multi-parameter and
non-complement-free, other than the relatively weak results known for
general valuations.


To make inroads on the problem of combinatorial auction design in the
presence of complements, we propose a  model for valuations with
complements that is parameterized by the ``size'' of the complements.
The model permits a succinct representation, a variety of
computationally efficient queries, and non-trivial
welfare-maximization algorithms and mechanisms.
Specifically, 
a {\em hypergraph-$r$ valuation~$v$} for a good
set~$M$ is represented by a hypergraph~$H=(M,E)$,
where every (hyper-)edge $e \in E$ contains at most~$r$ vertices and has a
nonnegative weight~$w_e$.
Each good $j \in M$ also has a nonnegative weight~$w_j$.
The value~$v(S)$ for a subset~$S \sse M$ of goods is defined as the
sum of the weights of the goods and edges entirely contained in~$S$.

We design the following polynomial-time approximation algorithms and
truthful mechanisms for welfare maximization with bidders with
hypergraph valuations.
\begin{enumerate}
\item
For bidders whose valuations
 correspond to  subgraphs of a known graph that is
 planar (or more generally, excludes a fixed minor),
we give a  truthful and $(1+\epsilon)$-approximate mechanism.

\item
We give a polynomial-time, $r$-approximation algorithm for welfare
maximization with hypergraph-$r$ valuations.  Our algorithm randomly
rounds a compact linear programming relaxation of the problem.

\item We design a different approximation algorithm and use it
to give a polynomial-time, truthful-in-expectation mechanism
that has an approximation factor of~$O(\log^r m)$.

\end{enumerate}
\end{abstract}







\newpage

\section{Introduction}

In a {\em combinatorial auction}, a set~$M$ of $m$ heterogeneous
and indivisible goods are allocated to~$n$ bidders with private
preferences.  We model the preferences of a bidder~$i$ via a {\em
  valuation}~$v_i$, a function from~$2^M$ to~$\mathbb{R}^+$ with
$v_i(\emptyset) = 0$ and $v_i(S)\le v_i(T)$ for any $S\sse T\sse M$.
The number $v_i(S)$ indicates bidder $i$'s ``maximum willingness to
pay'' for the set of goods~$S\sse M$.
Several real-world scenarios can be modeled as combinatorial auctions,
most famously FCC auctions for spectrum; see the
books~\cite{CSS06,milgrom} for many more details and examples.
One fundamental optimization problem in combinatorial auctions is {\em
welfare maximization}, which is the problem of allocating the goods to
the bidders to maximize $\sum_{i=1}^n v_i(S_i)$, where~$S_i$ denotes
the goods allocated to~$i$.

The difficulty of welfare maximization, and of combinatorial auction
design more generally, depends on the structure of bidders' valuations.
A key issue is whether valuations exhibit {\em complements}, where one
good (e.g., a left shoe) has additional value in conjunction with
another good (a right shoe).  Formally, a {\em complement-free} (or
{\em subadditive}) valuation~$v$ satisfies $v(S) + v(T) \ge v(S \cup
T)$ for every pair $S,T \sse M$ of bundles.  Most previous work
in theoretical computer science on welfare maximization for
combinatorial auctions has focused on complement-free valuations and
several of its interesting subclasses (see Section~\ref{ss:related}
for details).
A simple valuation~$v$ with complements is a {\em single-minded} one,
where there is a bundle $S \sse M$ of at least two goods and a
number~$c > 0$ such that $v(T)$ equals $c$ if $T \supseteq S$ and 0
otherwise.

Complements are prevalent in the motivating applications for
combinatorial auctions.
For example, in the FCC spectrum action there are natural synergies
between nearby licenses.
However, welfare maximization in
the presence of complements is notoriously difficult, and this
intractability has stymied theoretical progress in the area.
For example, there are no known positive results for combinatorial
auctions in which bidder valuations are multi-parameter and
non-complement-free, other than the relatively weak results known for
general valuations.
Even with mere
single-minded valuations, there is no polynomial-time approximation
algorithm for the welfare maximization problem with sub-polynomial
approximation ratio (assuming $P \neq NP$)~\cite{LOS,Sandholm02}.

The combined importance and intractability of valuations with general
complements motivates models for valuations with complements that are
restricted in some way.  Ideally, such a model should be parameterized
by the ``size'' of the complements, and should permit
a succinct representation, a variety of computationally efficient
queries, and non-trivial welfare-maximization algorithms.
In this paper we study such a model.

We study welfare maximization algorithms and truthful combinatorial
auctions for a natural (hyper)graphical valuation model that
concisely expresses limited complements.
Precisely, a {\em hypergraph-$r$ valuation~$v$} for a good
set~$M$ can, by definition, be represented by a hypergraph~$H=(M,E)$,
where every edge $e \in E$ contains at most~$r$ vertices and has a
nonnegative weight~$w_e$.
Each good $j \in M$ also has a nonnegative weight~$w_j$.\footnote{No
  interesting approximation results are possible if weights are
  allowed to be negative (assuming $P \neq NP$).}
The value~$v(S)$ for a subset~$S \sse M$ of goods is defined as the
sum of the weights of the goods and edges entirely contained in~$S$:
\begin{equation}\label{eq:def}
v(S) = \sum_{j \in S} w_j + \sum_{e \,:\, e \sse S} w_e.
\end{equation}
The parameter
$r$ is called the {\em rank} of the valuation $v$.
The rank of a set of valuations is the maximal rank of any valuation in the set.

For example, consider a {\em graph} (i.e., hypergraph-$2$) valuation.
A bidder with such a valuation has a base value for each good, and
also enjoys bonuses if it acquires certain  pairs of complementary
goods.
Even in this special case, each bidder possesses a polynomial number of
different private parameters.
Single-minded bidder valuations are essentially a single-parameter
special case of hypergraph valuations --- they correspond to the
hypergraph valuations with a single edge with non-zero weight.
A hypergraph-$r$ valuation can be described via $O(m^r)$ parameters,
which is polynomial in the number of goods when~$r$ is a fixed constant.

Summarizing, complements are an important, unavoidable aspect of
combinatorial auctions, and they have been underrepresented in
algorithmic mechanism design.  Hypergraph valuations are a natural
combinatorial representation of valuations that expresses complements
in a controllable way.
Even the most basic questions about welfare maximization
with bidders with hypergraph valuations
were completely open before the present work.
\begin{enumerate}

\item What is the best-possible polynomial-time approximation ratio
  for welfare maximization (assuming $P \neq NP$), as a function of
  the valuations' ranks?

\item Can equally good approximation ratios be achieved by
  incentive-compatible 
  auctions?

\item Does the combinatorial structure of the valuations permit novel
  combinatorial auction designs, for example by applying
graph-theoretic tools?

\end{enumerate}

\subsection{Our Results}

We contribute to all three of the questions above.  Our main results are
following.
\begin{enumerate}

\item
We use tools from graph theory to give a
 truthful and $(1+\epsilon)$-approximate
 mechanism
 for bidders with valuations that correspond to
 subgraphs of a known graph that is planar (or more generally, excludes a fixed minor).
Such valuations plausibly model some interesting combinatorial
 auctions, such as those where the predominant synergies involve
 ``neighboring'' goods (e.g., spectrum licenses).

 Our technical approach here is to use Baker-type graph decompositions
 to design
 mechanisms that reduce the given welfare maximization
 problem to one with valuations supported by a graph with bounded
 treewidth, which we solve optimally in polynomial time using dynamic
 programming and the VCG mechanism.

\item
We give a polynomial-time $r$-approximation algorithm for welfare
maximization with
general hypergraph-$r$ valuations.  This guarantee is nearly the best
possible, assuming $P \neq NP$.
Our algorithm uses a correlated randomized rounding of 
a compact linear programming relaxation of the problem.

\item
With rich multi-parameter valuations like hypergraph valuations, only
highly restricted types of approximation algorithms lead to truthful
mechanisms; in particular, the $r$-approximation algorithm above does not.
We instead design another algorithm, which
randomly rounds the exponentially-sized configuration linear
programming relaxation, and prove that it is
``maximal-in-distributional-range (MIDR)'' and hence induces a
truthful-in-expectation mechanism.  Our mechanism runs in polynomial
time and has an approximation factor of $O(\log^r m)$ to the social
welfare, for general
hypergraph-$r$ valuations.  For small~$r$, this guarantee is
much better than the
$O(\sqrt{m})$ approximation factor achieved by combinatorial auctions
for general valuations, the only previous guarantee known for
multi-parameter non-complement-free valuations.

\end{enumerate}

\subsection{Related Work}
\label{ss:related}

\eat{
these are the points by Tim, should be removed when we are done (Tim?).
\begin{itemize}
\item
some general combinatorial auction references.  LOS paper.
 econ (milgrom?) papers that try address complements, though not
 with approximation.

\item
previous work on the model:
papers on hypergraph valuations: conitzer et al, plus one other ai
paper.  discuss what they did, no approximation;
also: what bidding language (OR, etc.) does it most closely
  correspond to? (check nisan, sandholm surveys)

\end{itemize}
} 
Combinatorial auctions
are
a paradigmatic problem in Algorithmic Mechanism Design~\cite{NR01},
as they are representative of problems in which
incentives and computational constraints clash.
For background on truthful approximation
mechanisms for combinatorial auctions we refer the reader
to~\cite{CSS06,BN07}. As the literature on the subject is very large
we only survey papers that are closely related to ours.

The most basic form of a valuation that exhibits complements is that of
single minded bidders, in which each bidder desires one specific
bundle.
Lehmann, O'Callaghan and Shoham~\cite{LOS} have presented a truthful
mechanism that is an $O(\sqrt{m})$-approximation when bidders are
single minded. It has also been shown~\cite{Sandholm02,LOS} that this
is essentially the best possible (unless $P=NP$), even without
incentive constraints.
Lavi and Swamy~\cite{LS05} have presented a randomized mechanism that
is truthful-in-expectation, and achieves $O(\sqrt{m})$-approximation
for general valuations.
While this approximation is essentially the best possible for general
valuations, it is obviously not attractive when $m$ is large.
Many previous works have sought improved approximations for restricted valuations; see~\cite{BN07} for a survey.  These have focused on subadditive valuations, and subclasses of them like submodular and ``XOS'' valuations.
As far as we know, this paper is the first to improve 
the $O(\sqrt{m})$-approximation results for valuations that exhibit
any form of
complements.

Conitzer, Sandholm and Santi~\cite{CSS05} have considered ''$k$-wise
Dependent Valuations'' which are similar to hypergraph-$k$ valuations
(except they also allow the edges' weights to be negative) and have
shown that the problem of finding the efficient allocation is NP-hard
even for graphs ($k=2$), and have also studied the problem of
eliciting such valuations.
Chevaleyre et al.~\cite{CEEM08} consider the same class of valuations
(which they call ``$k$-additive'') and prove similar computational
hardness results, and also study negotiation protocols between the
agents.
Unlike these papers our focus is on approximating the welfare-maximizing
allocation and designing truthful mechanisms with good
approximations.


Each of our technical results includes some ingredients that were
developed in earlier works.  Our $r$-approximation algorithm for
welfare maximization (Theorem~\ref{t:wm}) is based on a sequence of
randomized rounding algorithms devised for the multiway cut
problem~\cite{CKR00}, the metric labeling problem~\cite{KT02},
graph homomorphism problems~\cite{LRS06}, and other
decomposition problems~\cite{KR}.  Our $O(\log^r m)$-approximate
truthful-in-expectation combinatorial auction uses a number of ideas.
The configuration linear programming relaxation, and the fact that
its tractability reduces to that of demand queries, is folklore (see
e.g.~\cite{BN07}).  The modified version we use with ``proxy bidders''
is based on an approach proposed by Feige~\cite{F06}, and recently used by Dobzinski, Fu, and Kleinberg~\cite{DFK10} to obtain truthful mechanisms.
To round this linear program we apply the decomposition technique of Lavi and
Swamy~\cite{LS05}, which builds on Carr and Vempala~\cite{CV02}.
To argue truthfulness (in expectation) we use the concept of a
``maximal-in-distributional-range (MIDR)'' algorithm, which was first
articulated by Dobzinski and Dughmi~\cite{DD09}, generalizing the
``maximal-in-range (MIR)'' definition in Nisan and Ronen~\cite{NR07}.
Finally, for our truthful approximation scheme for valuations
supported by an excluded-minor graph, uses the decomposition theorem of
Devos et al.~\cite{DDOPRSV04}, which in turn uses the structure theorem of
Robertson and Seymour \cite{RS04}.

\section{Mechanism Design Preliminaries}
\label{s:prelim}

We study direct-revelation mechanisms for players with quasi-linear
utilities.  In this setting, a {\em mechanism} is a
pair~$(\allocs,\prices)$ of algorithms.  The {\em allocation
rule~$\allocs$} takes as input a reported hypergraph valuation~$b_i$
from each
player, and outputs a feasible assignment of goods to the players.
The {\em payment rule}~$\prices$ takes the same input and computes a
payment to the mechanism from each player.  If a player~$i$ with a
valuation~$v_i$ is assigned the goods~$\alloci(b_1,\ldots,b_n) \sse M$
and payment~$\pricei(b_1,\ldots,b_n) \ge 0$, then $i$ earns {\em
utility} $v_i(\alloci(b_1,\ldots,b_n)) - \pricei(b_1,\ldots,b_n)$.
Let $x=X(v_1,\ldots,v_2)$ be the allocation when all players report truthfully.
The {\em welfare} $v(x)$ of the allocation $x=(x_1,\ldots,x_n)$
with respect to the
valuations $v=(v_1,\ldots,v_n)$ is defined to be $v(x)=\sum_{i=1}^n v_i(x_i)$.

A deterministic mechanism is {\em truthful} if, for every player~$i$
and fixed reports $b_{-i}$ by the players other than~$i$, player~$i$
maximizes its utility by reporting its true valuation~$v_i$,.

The {\em VCG mechanism} is a truthful mechanism that achieves the
maximum-possible welfare (but generally not in polynomial
time~\cite{NR01}).  This mechanism computes the allocation that
maximizes the welfare with respect to the reported valuations, and
then charges suitable ``externality payments'' to achieve truthfulness.
A {\em maximal-in-range (MIR)} allocation rule pre-commits to a subset
of feasible allocations, before receiving players' reports, and
maximizes the welfare with respect to the reported valuations over
this subset.  Every MIR allocation rule induces a deterministic
truthful mechanism via VCG-type payments~\cite{NR07}.  MIR allocation
rules are useful when the subset of allowable allocations can be
made structured enough to permit polynomial-time optimization, yet
large enough to always contain a near-optimal allocation.
A {\em maximal-in--distributional-range (MIDR)} allocation rule is the
following randomized analogue: pre-commit to a set of distributions
over feasible allocations, prior to receiving players' reported
valuations; choose the distribution that maximizes expected welfare
with respect to the reported valuations; and finally, sample a single
feasible allocation from the chosen distribution.  Every MIDR
allocation rule induces a truthful-in-expectation mechanism when
supplemented with VCG-type payments~\cite{DD09}.


\section{Truthful Approximation Mechanisms for Excluded-Minor Graphs}
In this section we impose restricted structure on bidders' hypergraph
valuations and rewarded with a near-optimal truthful mechanism.
The key assumption is that bidders valuations have rank~2, with the
edges drawn from a known graph that excludes a fixed minor, such as a
planar graph.  This assumption is not unreasonable in the motivating
applications for combinatorial auctions.
For example, Kagel, Lien, and Milgrom \cite{KLM10} have modeled
player valuations in the FCC spectrum action as subgraphs of a fixed
planar graph that captures the synergies of winning nearby licenses
(similar models were suggested in Brunner et al. \cite{BGHL10}).
In \ShortLong{
the full version}{Appendix~\ref{app:chromatic}}
we also present a truthful mechanism with approximation factor equal to
the chromatic number of the square of the line graph of the graph.

\subsection{Standard Definitions}
We begin by recalling some standard definitions.
A graph $X$ is a \emph{minor} of a graph $G$ if there exists a sequence of edge deletions and edge contraction that start with $G$ and end with $X$. A graph $G$ \emph{excludes $X$ as a minor} if $X$ is not a minor of $G$.

\ShortLong{
Due to limited space we omit 
the standard definition of \emph{tree decomposition} with width $k$ of a graph $G$, it can be found in the full version.}{
\begin{definition}
A graph $G=(M,E)$ has a \emph{tree decomposition} with width  $k$ if there exists a tree  $(X, T)$, where $X = \{X_1, ..., X_n\}$ is a set of \emph{bags} (where each $X_i \subseteq M$), and $T$ is a tree whose nodes are the bags $X_i$. Such that
\begin{enumerate}
\item  $\bigcup {X_i} = M$ (every node in $M$ appears in at least one bag);

\item $\max |X_i| \le k+1$ (each bag contains at most $k+1$ nodes);

\item for every $(u,v) \in E$ there exists $i$ such that $u,v \in X_i$ (for every edge there exists a bag that contains both nodes);

\item if $X_i$, $X_j$ and $X_\ell$ are nodes, and $X_\ell$ is on the path in $T$ from $X_i$ to $X_j$, then $X_i \cap X_j \subseteq X_\ell$ (for every node in $M$, the set of bags that contain it induces a subtree of $T$).
\end{enumerate}
\end{definition}
}
The \emph{treewidth} of a graph $G$ is the least integer $k$ for which $G$ has a tree decomposition with width  $k$.

A graph valuation $(M,E_v)$ of a player is a {\em subgraph} of a given graph $G=(M,E)$ of every edge $e\in E_v$ belongs to $E$, that is,  $e\in E_v$ implies  $e\in E$.

\subsection{Mechanisms}
We next show that if all players' valuations are subgraphs of a given graph with {\em constant} treewidth then welfare maximizing allocation, and thus the efficient and truthful VCG mechanism, can be computed in polynomial time.

\begin{theorem}\label{thm:treewidth}
If the valuations of all players are subgraphs of a given graph $G$ with constant treewidth $k$,
then the welfare maximizing allocation can be computed in time $n^{O(k)}$.
\end{theorem}
\ShortLong{
See the full version for the proof of the theorem.
}{
To prove the theorem, let $(T,X)$ be a tree decomposition of $G$, such that all bags are of size $\le k+1$. By Bodlaender \cite{B96} such a tree decomposition can be computed in linear time.
Let $v_\ell$ be the valuation of player $\ell \in N$. For any bag $X_i \subseteq M$ let $A':X_i \to N$ be a allocation of items in $X_i$ (we call $A'$ a partial allocation).
Given two (partial) allocations $A_1,A_2$ and a set $Y \subseteq M$ we say that $A_1$ \emph{agrees} with $A_2$ on $Y$ and write $A_1(Y)=A_2(Y)$ if $\forall y \in Y, A_1(y)=A_2(y)$.
We prove a slightly stronger statement using a standard dynamic programming induction argument.

\begin{lemma}\label{lem:treewidth-strong}
For any bag $X_i$ and any partial allocation $A':X_i \to N$, one can find $$\operatorname*{argmax}_{A \mid A(X_i)=A'(X_i) } \left \{ \sum_{\ell \in N} v_i(A^{-1}(\ell))  \right \}$$ the maximum efficient allocation that agrees with $A'$ on $X_i$ in time $n^{O(k)}$.
\end{lemma}
\begin{proof}
Consider the tree $T$ and fix an arbitrary root bag $X_r$.  For any bag $X_j$ let $M[X_j] \subseteq M$ be the set of all nodes that are in bags that are in the subtree of $T$ that contains $X_j$ and all its descendant bags (with respect to the root $X_r$).

Consider any two immediate children bags $Y_1,Y_2$ of a bag $X_i$ and let $M_j =M[Y_j] \setminus X_i$ for $j \in \{1,2\}$. The main observation is that there does not exists a pair $(u,v) \in M \times M$ such that $u \in M_1$ and $v \in M_2$ (this includes the case that $u=v$).
This is true because by the property of a tree decomposition, such an edge (or vertex if $u=v$) would imply that both $u$ and $v$ must belong to $X_i$, but we removed $X_i$ from both $M_1$ and $M_2$.

This observation suggests that the optimal allocation for $M[X_i]$ has the property that its restriction to $M[Y_j]$ is also the optimal solution for subgraph induced by $M[Y_j]$ (for any $j \in \{1,2\}$). This implies that we can use a standard dynamic programming algorithm. Given a bag $X_i$ we can compute for each allocation $A':X_i \to N$ (there are at most $n^{k}$ such allocations), the optimal allocation $A:M[X_i] \to N$ that agrees with $A'$ (so $A'(X_i)=A(X_i)$).  This can be done by going over all immediate children bags of $X_i$ ($Y_1,\dots,Y_\ell$) and for each such bag $Y_j$ going over all partial allocations $A'_j:Y_j \to N$ that agree with $A':X_i \to N$ and choose the one whose optimal allocation $A_j: M[Y_j]\to N$ for subgraph induced $M[Y_j]$ has maximal welfare. By the observation above, this procedure will find the optimal allocation for $M[X_i]$ that agrees with any partial allocation $A':X_i \to N$. Moreover this can be done using $n^{k+O(1)}$ space and time using a standard dynamic programming approach.
\end{proof}
}

Using the result on bounded treewidth graphs, we devise an approximately efficient, truthful mechanism for valuation that are subgraphs of graphs excluding a fixed minor. We begin by presenting the result for planar graphs on which the basic tool we use is Baker's decomposition:

\begin{theorem}\cite{Baker94}
For any planar graph $G=(M,E)$ and parameter $k$ there exists a partition of $G$ into $k+1$ parts $P_0,\dots,P_k$ such that for any $0 \le i \le k$ the graph induced by $M\setminus P_i$ has treewidth $\le 3k$.
\end{theorem}

As shown in \cite{Baker94} a simple way to generate the elements of the partition is to choose an arbitrary root $r$ and for any $0 \le i \le k$ let $P_i$ be the set of all points whose distance from $r$ is $(k+1) \ell + i$ for some integer $\ell$.

\begin{figure}[t]
\hrule\medskip
\textbf{Common Knowledge}: a planar graph $G=(M,E)$, parameter $\epsilon>0$.\\
\textbf{Input}: reported graph valuations (that are subgraphs of $G$)
$v_1=(M,E_1,w_1),\ldots,v_n=(M,E_n,w_n)$.
\begin{enumerate}

\item Let $k =  \lceil 2/\epsilon \rceil$.  Choose an arbitrary root $r$ and for any $0 \le i \le k$ let $P_i$ be the set of all points whose distance from $r$ is $(k+1) \ell + i$ for some integer $\ell$.

\item For every $0 \le i \le k$ let $M_i$ be the graph induced by $M \setminus P_i$. Note that $M_i$ is $k$ outer-planar and hence has tree width $\le 3 k$.

\item For every $0 \le i \le k$, compute the efficient allocation for $M_i$ using the valuations that are induced by these items.
(can be done in polynomial time by Theorem~\ref{thm:treewidth}). Let $A_i$ be the resulting allocation, and let $R_i$ be the welfare of $A_i$. 

\item Pick any $i^*\in \operatorname*{argmax}_i\{R_i\}$, then use the allocation $A_{i^*}$.


\end{enumerate}
\caption{\textsf{A maximal-in-range $(1+\epsilon)$-approximation algorithm when all players' valuations are subgraphs of a given planar graph}\label{f:planar}}

\medskip\hrule
\end{figure}

\begin{theorem}\label{thm:planar}
Fix any $\epsilon>0$. Assume that we are given a planar graph $G$ such that the valuations of all the players are subgraphs of $G$.
The algorithm of Figure \ref{f:planar}, when combined with VCG payments, yields a truthful, polynomial time, $(1+\epsilon)$-approximation mechanism for the welfare maximization problem.
\end{theorem}
\begin{proof}
Let $P_0,\dots,P_k$ be the parts. For each graph $G_i$ induced by the vertices $M \setminus P_i$, let $A_i$ be the allocation that is maximally efficient for $G_i$ among all allocations that are restricted to items $M \setminus P_i$.
By Theorem~\ref{thm:treewidth} $A_i$ can be computed in polynomial time.
Note that the allocation rule that computes $A_i$ is maximal-in-range, and therefore the algorithm that chooses the best $A_i$ is also maximal in range. Maximal in range algorithms can be combined with truth-telling VCG payments, computable in polynomial time, to yield a truthful mechanism.

Consider an optimal allocation $A$. For any item $u \in M$ let $\beta_u$ be the contribution of $u$ to the welfare induced by $A$, for any edge $e=(u,v) \in E$ let $\beta_e$ be the contribution of edge $e$ to the welfare induced by $A$. Let $\beta = \sum_{u \in M} \beta_u + \sum_{e\in E} \beta_e$ be the optimal welfare obtained by $A$. For any part $P_i$, let $\alpha_i$ be the welfare of $A$ from the items in $P_i$ and the edges that have at least one vertex in $P_i$ (formally $\alpha_i = \sum_{u \in P_i} \beta_u + \sum_{e \in E \mid e \cap P_i \ne \emptyset} \beta_e$). Observe that $\sum_{0 \le i \le k} \alpha_i \le 2 \beta$ (since each vertex is counted once and each edge is counted at most twice). A mechanism that optimizes on the graph induced by $M\setminus P_i$ obtains welfare of at least $\beta-\alpha_i$.
Thus the welfare of choosing the part $P_i^*$ with the highest welfare is at least $\frac{1}{k+1} \sum_{0 \le i \le k} \beta-\alpha_i \ge \beta - \frac{2}{k}\beta \ge \beta (1-\epsilon)$.
\end{proof}

We can extend Theorem \ref{thm:planar} to a much larger family of graphs excluding any fixed minor using the following powerful decomposition theorem.
\begin{theorem}\cite{DDOPRSV04}
For any graph $X$ and parameter $k$ there exists a constant $c=c(X,k)$  such that the following holds.
For any graph $G=(M,E)$ that excludes $X$ as a minor there exists a partition into $k+1$ parts $P_0,\dots,P_k$ such that for any $i$ the graph induced by $M\setminus P_i$ has treewidth $c$.
\end{theorem}
Moreover, the proof in \cite{DDOPRSV04} (along with the required structure theorem) implies that the partition can be computed in polynomial time.
Using this and the same approach as in Theorem \ref{thm:planar}
gives the following result.
\begin{proposition}
Fix any $\epsilon>0$ and a graph $X$. Assume that we are given a graph $G$ that excludes $X$ as a minor such that the valuations of all the players are subgraphs of $G$.
There exists a truthful mechanism that runs in polynomial time and
guarantees a $(1 + \epsilon)$-approximation to the social welfare.
\end{proposition}

\section{Approximate Welfare Maximization Algorithm}
\label{s:wm}

This section gives an $r$-approximation algorithm for
welfare maximization with general and private hypergraph-$r$
valuations.  We do not expect
that this algorithm, or small modifications to it, can lead to a
truthful 
auction (see Remark~\ref{rem:wm}).
Nonetheless, this result makes two important points.  First, in
conjunction with known hardness results (Theorem~\ref{t:hard}), it
precisely pins down the approximability of the welfare maximization
problem for this class of bidder valuations.
Second,
it demonstrates
that welfare maximization with hypergraph-$r$ valuations is as
tractable as with the far less expressive class of single-minded
valuations with bundle size at most~$r$.

Recall that an instance of the welfare maximization problem is
described by~$n$ hypergraphs~$H_1=(M,E_1,w_1),\ldots,H_n=(M,E_n,w_n)$ of
rank~$r$, which
represent players' valuations, all with a common vertex set~$M$, which
represents the goods.
Feasible solutions are assignments of the goods
to the players, which correspond to
an ordered partition~$S_1,\ldots,S_n$ of~$M$.  The goal is to compute the
assignment maximizing the welfare $\sum_{i=1}^n v_i(S_i)$, where the
hypergraph~$H_i$ defines the value~$v_i(S_i)$ as in
equation~\eqref{eq:def}.
We next present the main result of this section.
\begin{theorem}\label{t:wm}
There is a polynomial-time, $r$-approximation randomized algorithm for welfare maximization with hypergraph-$r$ valuations.
\end{theorem}
Our result almost matches the lower bound for the algorithmic problem.


\begin{theorem}[\cite{T01}]\label{t:hard}
There is no $r/2^{O(\sqrt{\log r})}$ approximation algorithm\footnote{In particular, for any $\epsilon>0$ there is no $r^{1-\epsilon}$ approximation algorithm.}
for welfare maximization with hypergraph-$r$ valuations, unless $P = NP$.
\end{theorem}
Theorem~\ref{t:hard} follows easily from hardness results for finding
independent sets in bounded-degree graphs.  It even applies to the
special case of single-minded valuations with desired bundles of size
at most~$r$, which correspond to hypergraph-$r$ valuations with only
one edge with non-zero weight.

We next describe our approximation algorithm.  Algorithms that use
only ``local'' information to make allocation decisions, like natural
greedy algorithms, do not seem capable of achieving an approximation
ratio that depends only on the rank~$r$.\footnote{This contrasts with
  single-minded valuations with bundle size at most~$r$, for which a
  simple greedy algorithm provides an~$r$-approximation.  For example,
  consider an instance with~2 players and $m$ goods.  The first
  player has value~$\sqrt{m}$, say, for the first good, and no value
  for anything else.  The second player's valuation is a star centered
  on the first good, where each vertex has weight~0 and each edge has
  unit weight.  The optimal solution gives all goods to the
  second player and has welfare~$m-1$, while most natural greedy
  algorithms allocate the first good to the first player and achieve
  welfare only~$\sqrt{m}$.}  Our algorithm solves and randomly rounds
the following linear programming relaxation of the problem.

\begin{eqnarray}
\label{eq:lp1}
\mbox{max}\ \ \  \   \sum_{i=1}^n & \left( \sum_{j \in M} w_{ij}x_{ij} + \sum_{e \in
  E_i} w_{ie}z_{ie} \right)\\
\label{eq:lp2}
\mbox{subject to:}&  \sum_{i=1}^n x_{ij} = 1 &\mbox{for every  good~$j$.}\\
& z_{ie} \le x_{ij} & \mbox{for every player~$i$, edge $e \in E_i$,}\notag\\
  &&\mbox{and good $j \in e$.} \label{eq:lp3}\\
& x_{ij} \ge 0 & \mbox{for every player~$i$ and good~$j \in M$.}\\
\label{eq:lp5}
& z_{ie} \ge 0 & \mbox{for every player~$i$ and edge~$e \in E_i$.}\\\notag
\end{eqnarray}

Every feasible assignment naturally induces a 0-1 feasible solution
to~\eqref{eq:lp1}--\eqref{eq:lp5} with equal objective function value,
where $x_{ij} = 1$ if and only if player~$i$ is assigned good~$j$, and
$z_{ie} = 1$ if and only if player~$i$ is assigned every good in~$e$.
The size of the linear program~\eqref{eq:lp1}--\eqref{eq:lp5} is
polynomial in the input size, so it can be solved in polynomial time.

We round the optimal solution to the linear
program~\eqref{eq:lp1}--\eqref{eq:lp5} using the randomized algorithm
shown in Figure~\ref{f:wm}.  The algorithm uses only the $x$-variables
to make assignments; the~$z$-variables are used in the analysis.


\begin{proof}{\bf (of Theorem~\ref{t:wm})}
Since by equality~(\ref{eq:lp2}) every good is fully assigned in $x^*$, the algorithm runs in polynomial time (both in expectation and with high probability). We next prove the approximation.
We claim that, for every player~$i$ and edge~$e \in E_i$, player~$i$
is assigned every good in~$e$ by the algorithm in Figure~\ref{f:wm}
with probability at least $z^*_{ie}/|e|$.
Similarly, each player~$i$ is assigned good~$j$ with probability at
least $x^*_{ij}$.  This claim, combined with linearity of expectation and
the assumption that $|e| \le r$ for all~$i$ and $e \in E_i$, implies
the theorem.

We observe that if a player~$\ell$ is chosen
in an iteration of the randomized rounding algorithm and no good of an
edge~$e$ has been assigned in a previous iteration, then this player
will get at least one
good of~$e$ with probability $\max_{j \in e} x^*_{\ell j}$, and every
good of~$e$ with probability $\min_{j \in e} x^*_{\ell j}$.

Now fix a player~$i$ and an edge~$e$.  (The same argument works for a
good~$j$, viewed as an edge of size~1.)
We only need to lower bound the probability that~$i$ receives all
of the goods in~$e$ in the same iteration of the randomized rounding
algorithm.


Consider the first iteration~$q$ that assigns at
least one good of~$e$ to some player.
Using Bayes' rule and the fact
that a single player is chosen uniformly at random each iteration,
the probability that~$i$ is the chosen player in iteration~$q$ is
\begin{equation}\label{eq:wm1}
\frac{\max_{j \in e} x^*_{ij}}{\sum_{\ell=1}^n \max_{j \in e} x^*_{\ell j}}
\ge
\frac{\max_{j \in e} x^*_{ij}}{\sum_{\ell=1}^n \sum_{j \in e} x^*_{\ell j}}
=
\frac{\max_{j \in e} x^*_{ij}}{\sum_{j \in e} \sum_{\ell=1}^n x^*_{\ell j}}
=
\frac{\max_{j \in e} x^*_{ij}}{|e|},
\end{equation}
where the final equality follows from the linear program
constraints~\eqref{eq:lp2}.

The probability that player $i$ receives every good in $e$ in round $q$, conditioned on $i$ being chosen in iteration $q$ (and hence receiving at least one good of $e$ at this round), is
\begin{equation}\label{eq:wm2}
\frac{\min_{j \in e} x^*_{ij}}{\max_{j \in e} x^*_{ij}}
\ge
\frac{z^*_{ie}}{\max_{j \in e} x^*_{ij}},
\end{equation}
where the inequality follows from the
linear program constraints~\eqref{eq:lp3}.  Combining
inequalities~\eqref{eq:wm1} and~\eqref{eq:wm2} shows that, independent
of the value of~$q$, player~$i$ is assigned every good in~$e$ with
probability at least $z^*_{ie}/|e| \ge z^*_{ie}/r$.
This completes the proof.
\end{proof}

\begin{figure}[t]
\hrule\medskip
\textbf{Input}: an optimal solution~$(x^*,z^*)$ of the linear
program~\eqref{eq:lp1}--\eqref{eq:lp5}.
\begin{enumerate}

\item While there is at least one unassigned good:

\begin{enumerate}

\item Choose a player~$i$ uniformly at random.

\item Choose a threshold $t \in [0,1]$ uniformly at random.

\item Assign to~$i$ every unassigned good~$j$ with $x^*_{ij} \ge t$.

\end{enumerate}

\end{enumerate}
\caption{\textsf{The randomized rounding algorithm for the linear
    program~\eqref{eq:lp1}--\eqref{eq:lp5}.}\label{f:wm}}
\medskip\hrule
\end{figure}

\begin{remark}
\label{rem:wm}
Turning the algorithm in Theorem~\ref{t:wm} into a truthful mechanism
with a similar approximation ratio seems difficult.  Essentially, the
only general technique for designing (randomized) truthful mechanisms
for multi-parameter valuations like hypergraph valuations
is via ``maximal-in-distributional-range (MIDR)'' approximation
algorithms, which we define in the next section.  The algorithm in
Theorem~\ref{t:wm} is not MIDR.  There are several reasons for this;
perhaps the most fundamental one is that hyperedges of different sizes
lose different approximation factors in~\eqref{eq:wm1}.
The decomposition technique of Lavi and Swamy~\cite{LS05}
cannot be used to transform it into an MIDR algorithm,
as our algorithm does not provide an approximation guarantee for
hypergraphs with negative edge weights.  (Negative edge weights cannot
be ignored or rounded up to zero without changing the problem.)  The
convex rounding technique of Dughmi, Roughgarden, and Yan~\cite{DRY11}
seems to yield interesting results only for subclasses of
complement-free valuations, and not for valuations with complements as
studied here.
\end{remark}


\section{A Truthful Approximation Mechanism}\label{s:ca}
In this section we present our truthful, $O(\log^r m)$-approximation mechanism for rank $r$ valuations.
Section~\ref{ss:demand} proves that demand oracles
can be implemented in polynomial time for
hypergraph valuations.  Section~\ref{ss:main} describes and analyzes
our truthful approximate combinatorial auction for hypergraph valuations.

\subsection{Demand Oracles and Hypergraph Valuations}
\label{ss:demand}

The first step of our algorithm, described in the next section, solves
an exponential-size linear programming relaxation of the welfare
maximization problem.
The complexity of solving this linear program
reduces, by dualizing and the ellipsoid method, to that of a {\em
  demand oracle} for a valuation~$v$: given a price~$p_j$ for each
good~$j \in M$, compute a bundle $S \in \arg\max \{ v(S) - \sum_{j \in
  S} p_j \}$.
Such queries can be answered in polynomial time for
hypergraph valuations.

\begin{proposition}\label{p:demand1}
Given a valuation~$v$ represented as a hypergraph $H=(M,E,w)$ and
prices~$p$ for the goods of~$M$, a bundle $S \sse M$ that maximizes
$v(S)-\sum_{j \in S} p_j$ can be computed in polynomial time.
\end{proposition}

\begin{proof}
(Sketch.)
First, a {\em value query} --- given a subset $S \sse M$ of goods,
return the value~$v(S)$ --- can be computed in time polynomial in the
size of~$H$ by brute force, using equation~\eqref{eq:def}.

Next, recall that a function~$v$ is supermodular if and only if,
  for every $S \sse M$ and $j,k \in M \sm S$,
$v(S \cup \{j,k\}) - v(S \cup \{j\}) \ge v(S \cup \{k\}) - v(S)$.
Since edge weights are nonnegative, the definition in equation~\eqref{eq:def} of a
  hypergraph valuation~$v$ easily implies that~$v(S)$, and hence also $v(S)
  - \sum_{j \in S} p_j$, is a supermodular function on~$M$.

Finally, recall that every supermodular function can be maximized
using a polynomial number of value queries
(e.g.~\cite{fleischer_survey}).
\end{proof}

\subsection{Description and Analysis of the Mechanism}
\label{ss:main}
The main result of this section is the following theorem
\begin{theorem}
Fix any constant $r$. There is a polynomial time, truthful-in-expectation mechanism that guarantees
$O(\log^r m)$-approximation to the social welfare when all players have hypergraph-$r$ valuations.
\end{theorem}

We now describe our new allocation rule.  Like the welfare
maximization algorithm in Section~\ref{s:wm}, it randomly rounds a
linear programming relaxation.
To obtain both truthfulness-in-expectation and a good approximation, we make two non-trivial changes.
First, our linear program fractionally allocates each good  at most $B$ times, where $B=\Omega(\log m)$.
On one hand, this duplication of the items results in a linear program with  a constant integrality gap.
On the other hand, rounding the linear program involves more ``conflict resolution''.  We define a natural rounding procedure that degrades the objective value by a factor of at most $O(B^r)$. The objective function of our linear program will involve modified valuation profiles, first proposed by Feige~\cite{F06} and termed \emph{proxy valuations} by
Dobzinski, Fu, and Kleinberg~\cite{DFK10}. 
The proxy  valuations will take into account the rounding procedure, and guarantee that the allocation rule  that first solves the linear program and then rounds the fractional solution is MIDR.

Our high-level algorithm is described in Figure~\ref{f:ca}.  In more
detail, the first step is to scale the edge weights of the reported
valuations $v_1,\ldots,v_n$ to obtain the corresponding proxy
valuations $v'_1,\ldots,v'_n$, as follows,
\begin{equation}\label{eq:mod}
v_i'(S) = \sum_{v \in S} \frac{w_{iv}}{ B} + \sum_{e \,:\, e \sse S} \frac{w_{ie}} {B^{|e|}}.
\end{equation}
We call~$v'_i$ the {\em proxy valuation} corresponding to~$v_i$.
Note that the proxy valuation of any set is exactly the expected value of the set when each item is assigned to the bidder independently with probability $1/B$.
The next step formulates the following linear program
\begin{eqnarray}
\label{eq:lp6}
\mbox{max} & f({\mathbf y}) = \sum_{i,S \ne \emptyset} v'_i(S) y_{i,S}&\\
\label{eq:lp7}
\mbox{subject to:} & \sum_{S\ne \emptyset} y_{i,S} \le 1 &\mbox{for every  player~$i$.}\\
\label{eq:lp8}
& \sum_i \sum_{S \mid j \in S} y_{i,S} \le B & \mbox{for every good~$j \in M$.}\\
\label{eq:lp9}
& y_{i,S} \ge 0 & \mbox{for every player~$i$ and bundle ~$S\sse M$.}
\end{eqnarray}
While the linear program has exponential size, it can be solved in polynomial time (by dualizing and using the ellipsoid algorithm)
given access to a demand oracle \cite{BN05iter}. By Proposition~\ref{p:demand1}, such a demand oracle can be computed in polynomial-time for hypergraph-$r$ valuations.
The next step is to use the Lavi-Swamy~\cite{LS05} approach to compute, in polynomial time, a distribution over integral solutions with polynomial size support with the following property. The expectation of  this distribution over integer solutions is equal to the optimal fractional solution,  divided by a universal constant $\alpha$. Any choice of $\alpha$ greater than the integrality gap of the linear program suffices; The integrality gap is $m^{\frac{1}{B+1}}$, which is bounded by a universal constant for any choice of $B=\Omega(\log m)$ (c.f.   \cite{LS05}).
In the fourth step we choose one of the integral solutions according to the prescribed distribution. Finally, in the last stage we resolve conflicts in the integral solution in such a way that if a player received a copy of the good in the integral solution then he is allocated that good with probability $1/B$, independently over all goods.

From the above description and Proposition~\ref{p:demand1} is it clear that the mechanism  runs in polynomial time. Next we prove the approximation guarantee and show that  the mechanism is MIDR.

\begin{figure}[t]
\hrule\medskip
\textbf{Parameters}: $B= \Omega(\log m)$. Universal constant $\alpha$ exceeding integrality gap of linear program~\eqref{eq:lp6}--\eqref{eq:lp9}.\\
\textbf{Input}: reported hypergraph-$r$ valuations
$v_1=(M,E_1,w_1),\ldots,v_n=(M,E_n,w_n)$.
\begin{enumerate}

\item Form proxy valuations $v'_i$ from $v_i$ for every player $i$. Let $v'=(v'_1,v'_2,\ldots,v'_n)$.

\item
Solve the linear program~\eqref{eq:lp6}--\eqref{eq:lp9} with respect to $v'$, obtaining
optimal solution~${\mathbf y^*}=\{ y^*_{iS} \}_{1 \le i \le n, S \sse M}$.

\item
Compute a decomposition~$\frac{{\bf y^*}}{\alpha} = \sum_{\ell \in \L}
\lambda_l {\bf z_{\ell}}$ into a convex combination of a polynomial
number of feasible integral solutions.

\item
Sample a feasible integral solution~${\mathbf z}=\{ z^*_{iS} \}_{1 \le i \le n, S \sse M}$ according to the
distribution $\{ \lambda_{\ell} \}_{\ell \in \L}$.

\item Independently for each good~$j=1,2,\ldots,m$: let $I_j=\{i \mid z_{i,S}=1, j \in S\}$, then with probability $\frac{|I_j|}{B}$ allocate $j$ to a uniformly random player in $I_j$, otherwise do not allocate $j$. 
    Note that  $|I_j|=\sum_i \sum_{S \mid j \in S} z_{i,S} \le B$ for every good $j$, thus $\frac{|I_j|}{B}\leq 1$.


\end{enumerate}
\caption{\textsf{The MIDR randomized rounding allocation rule for the linear
    program~\eqref{eq:lp6}--\eqref{eq:lp9}.}\label{f:ca}}
\medskip\hrule
\end{figure}

\begin{proposition}
The mechanism is MIDR and the approximation guarantee is at most  $\alpha B^r$
\end{proposition}
\begin{proof}
First, we prove that the mechanism is MIDR. Let $M(v)$ denote the expected social welfare achieved by the mechanism on receiving reports $v$.
Let ${\mathbf x}({\mathbf z})$ be the integral allocation of the mechanism (at the end of stage (5)) given a feasible integral solution ${\mathbf z}$ from stage 4.
Then from linearity of expectation and the fact that per player, the probability of getting allocated each good in step (5) is independent we have
\begin{equation}\label{eq:p1}
\expect{v({\mathbf x}({\mathbf z}))} = \sum_{i,S \ne \emptyset} v'_i(S) z_{i,S} = f({\mathbf z}).
\end{equation}
for any ${\mathbf z}$. By definition of the decomposition of stage 3, and using linearity of $f(.)$,  we have
\begin{equation}\label{eq:p2}
\expect [\{ \lambda_{\ell} \}_{\ell \in \L}] {f({\mathbf z}_\ell)} = f(\mathbf y^*)/\alpha.
\end{equation}
Thus, from equalities (\ref{eq:p1}),(\ref{eq:p2}) we conclude that
\begin{equation}\label{eq:p3}
M(v) = \expect [\{ \lambda_{\ell} \}_{\ell \in \L}] {\expect{v({\mathbf x}({\mathbf z_l}))}} = f(\mathbf y^*)/\alpha
\end{equation}

Since the linear program maximizes $f$ over its feasible domain, equation~(\ref{eq:p3}) implies that the mechanism optimizes $M$ over the set of distributions in its range.  Thus, the mechanism is MIDR.

We now prove the approximation guarantee. Let $x^*$ denote the vector encoding the  welfare-maximizing allocation for valuation profile $v$, let $v(x^*)$ denote the optimum welfare, and $v'(x^*)$ denote the corresponding ``proxy'' welfare. Observe that $f(y^*) \geq f(x^*)$ by optimality of $y^*$ for the linear program.  Since $x^*$ is a feasible integral allocation, $f(\mathbf x^*) =v'(\mathbf x^*)$. By equation (\ref{eq:mod}) we have that
\begin{equation}\label{eq:p4}
v'({\mathbf x}) \ge \frac{v({\mathbf x})}{B^r}
\end{equation}
for any integral allocation ${\mathbf x}$,  in particular, for ${\mathbf x^*}$. Combining with (\ref{eq:p3}),(\ref{eq:p4}) it follows that
\begin{equation}\label{eq:p5}
M(v) =  \frac{f(\mathbf y^*)}{\alpha} \ge \frac{v'(\mathbf x^*)}{\alpha} \ge \frac{v({\mathbf x^*})}{\alpha B^r}
\end{equation}
\end{proof}

\bibliographystyle{plain}
\bibliography{bib}

%
%
%
%
%
%
%

\ShortLong{
}{
\appendix
\section{When the Square of the Line Graph has a Small Chromatic Number}
\label{app:chromatic}
In this section we present a  truthful mechanism with approximation that equals to
the chromatic number of the square of the line graph of the common graph $G$ (when all valuations are subgraphs of $G$).
For example, if $G$ has bounded degree $d$ then the square of its line graph has chromatic number $O(d^2)$.

We begin with some standard definitions.
Given a graph $G=(V,E)$ let $L(G)=(E,\{(u,v),(w,x) \in E \mid | \{u,v\} \cap \{w,x\} \ne \emptyset\})$ be the \emph{line graph} of $G$ whose vertices are the edges of $G$ and nodes in $L(G)$ share an edge in $L(G)$ if the two corresponding edges in $G$ intersect at some vertex of $G$. Given a graph $G=(V,E)$ let $G^2=(V,\{(u,v) \mid \exists w, (u,w),(w,v) \in E \vee (u,v) \in E \})$ be the \emph{square graph} whose edges are all length one and two paths in $G$.

We are now ready to present our result.
\begin{proposition}
Let $G$ by any graph such that the valuations of all players are subgraphs of $G$.
Then there exists a truthful mechanism that runs in polynomial time and is a $\chi(L(G)^2)$-approximation.
\end{proposition}
\begin{proof}
One can color edges of the square graph with $\chi(L(G)^2)$ colors. So any two edges with the same color are at distance at least 2. We can iterate over all colors and for each color allocate only the items that are induced by the edges of that color. Since edges are at distance at least 2, we can compute the optimal allocation for each such edge independently.
\end{proof}

One simple corollary of the proposition is that if the common graph $G$ has bounded degree $d$ then
we can design a truthful mechanism with $O(d^2)$-approximation.
}
\ShortLong{
}{
}
\end{document}